\documentclass[11pt, a4paper]{article} 

\usepackage{a4wide}
\usepackage[utf8]{inputenc}
\usepackage[T1]{fontenc}
\usepackage[english]{babel} 
\usepackage{amsmath}
\usepackage{amsfonts}
\usepackage{amssymb}
\usepackage{amsthm}
\usepackage{euscript}
\usepackage{tikz}
\usepackage{url}

\newtheorem{theo}{Theorem}
\newtheorem{defi}[theo]{Definition}
\newtheorem{prop}[theo]{Proposition}

\newtheorem{lemma}[theo]{Lemma}

\newcommand{\N}{\mathbb{N}}
\newcommand{\F}{\mathbb{F}}

\newcommand{\HH}{\mathbf H}

\newcommand\mc[1]{\mathcal{#1}}

\newcommand\gen[1]{\langle #1\rangle}

\newcommand{\Siran}{\v{S}ir\'a\v{n} }

\DeclareMathOperator{\rank}{Rank}
\DeclareMathOperator{\Ker}{Ker}

\DeclareMathOperator{\syst}{syst}
\DeclareMathOperator{\csyst}{csyst}
\DeclareMathOperator{\Area}{Area}

\begin{document}

\title{Tradeoffs for reliable quantum information storage in surface codes and color codes}

\author{
    Nicolas Delfosse\\
    delfosse@lix.polytechnique.fr\\
    INRIA Saclay \& LIX, \'Ecole Polytechnique (UMR 7161), 91128 Palaiseau, France
}

\maketitle

\begin{abstract}
The family of hyperbolic surface codes is one of the rare families of quantum LDPC codes with non-zero rate and unbounded minimum distance.
First, we introduce a family of hyperbolic color codes. This produces a new family of quantum LDPC codes with non-zero rate and with minimum distance logarithmic in the blocklength.
Second, we show that the parameters $[[n, k, d]]$ of surface codes and color codes satisfy $k d^2 \leq C (\log k)^2 n$, where $C$ is a constant that depends only on the row weight of the parity--check matrix. Our results prove that the best asymptotic minimum distance of LDPC surface codes and color codes with non-zero rate is logarithmic in the length.
\end{abstract}

\section{Introduction}



The generalization of Low Density Parity--Check (LDPC) codes to the quantum setting is far from evident. The fact that most of the constructions of quantum LDPC codes have bounded minimum distance illustrates this difficulty. Such a distance is not sufficient for a family of quantum LDPC codes.
Topological constructions of quantum LDPC codes \cite{Ki03, BK98, BM06, BM07a, BPT10, CDZ11, TZ09} allow us to obtain a growing minimum distance, but their dimension is generally small.
The construction of families of non-zero rate quantum LDPC codes with unbounded minimum distance is highly non trivial.
Only a small number of families satisfy these two conditions. For exemple, hyperbolic surface codes, introduced by Freedmann, Meyer and Luo \cite{FML02} or by Zémor \cite{Ze09}, achieve a non-vanishing rate with a minimum distance logarithmic in the length.

In this paper, we introduce a family of hyperbolic color codes using double covers of hyperbolic tilings. This produces a new family of quantum LDPC codes with non-zero rate and logarithmic minimum distance.
The second part of this paper is devoted to a bound on the parameters of surface codes and color codes. We prove that the parameters $[[n, k, d]]$ of these codes satisfy
$$
k d^2 \leq C (\log k)^2 n,
$$
where $C$ is a constant that depends only on the row weight of the parity--check matrix.
This tradeoff between the parameters of topological codes proves that we cannot do better than a logarithmic distance using these quantum LDPC codes with non-zero rate.
Our proof is based on results of Gromov in riemannian geometry \cite{Gr92}.
This inequality can be seen as a generalization of the Bravyi-Poulin-Terhal bound that is available only for square tilings \cite{BPT10}.

The definition of surface codes and color codes from the CSS (Calderbank, Steane, Shor) construction is recalled in Section~\ref{section:parameters}. Our family of hyperbolic color codes is introduced in Section~\ref{section:hyp_color_codes} and their parameters are computed. The Section~\ref{section:gromov} is devoted to the generalization of Bravyi-Poulin-Terhal tradeoff to every kind of tiling.

\section{Parameters of surface codes and color codes}
\label{section:parameters}

{\bf CSS codes:}
The family of CSS codes \cite{CS96, St96} relies classical and quantum codes. A {\em CSS code} is defined by two matrices $\HH_X \in \mc M_{r_Z, n}(\F_2)$ and $\HH_Z \in \mc M_{r_Z, n}(\F_2)$ such that every row of $\HH_X$ is orthogonal to every row of $\HH_Z$.
Denote by $C_X$ the kernel of the matrix $\HH_X$ and denote by $C_Z$ the kernel of the matrix $\HH_Z$.
The parameters of this CSS code are $[[n, k, d]]$ where $n$ is the length,  $k=n-\rank \HH_X - \rank \HH_Z$ is the number of encoded quantum bits (qubits) and the minimum distance is the minimum weight of a codeword $x$ which is in $C_Z \backslash C_X^\perp$ or in $C_X \backslash C_Z^\perp$.

\vspace{.2cm}
\noindent
{\bf Tiling of surface:}
A combinatorial description of surface is sufficient to define surface codes and color codes. To avoid confusion between the combinatorial structure of the surface, needed to define these codes, and its riemannian structure, that appears in the last section of this paper, we will talk about a tiling of surface or a tiling. A {\em tiling} is defined to be a triple $G=(V,E,F)$, where $(V,E)$ is a graph and $F$ is the set of faces defined by the embedding of $(V,E)$ in a compact connected 2-manifold (surface) without overlapping edges. A face is given as the set of edges on its boundary.
In what follows, we denote by $V=\{v_i\}_{i=1}^{|V|}$ the vertices, we denote by $E=\{e_i\}_{i=1}^{|E|}$ the edges and we denote by $F=\{f_i\}_{i=1}^{|F|}$ the faces of the tiling $G$.

\vspace{.2cm}
\noindent
{\bf Cycle codes:} A {\em cycle} of a graph is a set of edges containing an even number of edges around every vertex of the graph. The set of cycles of a graph is a binary linear code called the {\em cycle code} of the graph.
A cycle can also be regarded as a vector of $\F_2^{|E|}$. The cycle code is the kernel of the incidence matrix of the graph. Recall that the incidence matrix of a graph $(V,E)$ is the binary matrix $I \in \mc M_{|V|,|E|}(\F_2)$ such that $I_{i, j}=1$ iff $v_i$ is contained in the edge $e_j$.

\vspace{.2cm}
\noindent
{\bf Surface codes:}
The face matrix of the tiling is the matrix $J\in \mc M_{|F|, |E|}(\F_2)$ such that $J_{i, j}=1$ iff $f_i$ contains the edge $e_j$.
The {\em surface code} associated with a tiling $G=(V,E,F)$ is the CSS code defined by the matrices $\HH_X$ and $\HH_Z$ such that $\HH_X \in \mc M_{|V|,|E|}(\F_2)$ is the incidence matrix of the tiling and $\HH_Z$ is the face matrix of the tiling.
The orthogonality results from the fact that faces are cycles, so they are in the kernel of the incidence matrix.
The parameters of these quantum codes are well known \cite{BM07a}. The length is $n=|E|$. The dimension is $k=2-\chi$, where $\chi=|S|-|E|+|F|$ is the Euler characteristic of the tiling. When the tiling is orientable, it is $2g$, where $g$ is the genus of the tiling, because we know that $g$ satisfies $\chi=2-2g$.
The minimum distance $d$ is the smallest length of a cycle which is not a sum of faces in the tiling $G$ or in the dual tiling $G^*$. Here, the sum of faces is the symmetric difference of the corresponding edge sets. A sum of faces is also called a boundary.

\vspace{.2cm}
\noindent
{\bf Color codes:}
Color codes are also defined from a tiling of surface but qubits are placed on the vertices instead of the edges.
Let $G=(V,E,F)$ be a tiling.
We assume that $G$ is trivalent, that is every vertex has degree 3, and that faces of the $G$ can be 3-colored such that two faces sharing an edge do not wear the same color.
The {\em color code} associated with this trivalent tiling $G$ with 3-colorable faces is the CSS code defined by the matrices $\HH_X=\HH_Z=\HH \in \mc M_{|F|,|V|}(\F_2)$ such that $\HH_{i, j}=1$ iff the face $f_i$ contains the vertex $v_j$.
The length of the color code associated with $G$ is $n=|V|$, its dimension is $4-2\chi$. When this tiling is orientable, this number is $k=4g$.
Denote by $C$ the code $\Ker \HH$. The minimum distance of the color code is the minimum weight of a vector $x$ in $C \backslash C^\perp$.

Every vector $x \in \F_2^{n}=\F_2^{|V|}$ can be regarded as a subset of the vertex set $V$.
This geometrical point of view leads to a description of the code $C= \Ker \HH$ and its orthogonal $C^\perp$ in term of cycle codes of graphs. Similarly to the case of surface codes, we deduce from this description that if a vector is included in a sufficiently small ball, it is in $C^\perp$ and it does not appear in the computation of the minimum distance. This property will be useful to obtain lower bounds on the minimum distance. It is summarized in the following lemma:

\begin{lemma} \label{lemma:color_d_bound}
Denote by $C=\Ker\HH$ the classical code associated to a color code. Let $x$ be a vector of $C$. Denote by $F(x)$ the set of faces incident to the vertices of $x$. If $F(x)$ induces a planar graph then $x$ is in $C^\perp$.
\end{lemma}

The subgraph $F(x)$ can be planar with many connected components.
A cycle code description of color codes and a proof of this lemma based on the shrunk lattices \cite{BM06} is proposed in \cite{De12}.

\section{Color codes beating Bravyi, Poulin and Terhal bound}
\label{section:hyp_color_codes}

\subsection{Hyperbolic cayley graphs}

In this section, we recall the construction of hyperbolic tilings used by Zémor in \cite{Ze09, DZ10} to construct surface codes. These tilings are based on quotients of triangle groups introduced by \Siran \cite{Si00}.
In what follows, $\ell$ and $m$ are two integers such that $1/\ell+1/m<1/2$.
Consider a group $T(\ell, m)$, generated by two elements $a$ and $b$, with the presentation:
$$
T(\ell, m)= \gen{ a, b \ | \ a^2=b^\ell=(ab)^m=1 }.
$$
Such an infinite group can be realized by complex matrices \cite{Si00, Ze09}.
Denote by $\tau(\ell, m)$ the Cayley graph of the group $T(\ell, m)$ and the generating set $S = \{a, b, b^{-1} \}$. It is the graph of vertex set $V=T(\ell, m)$ with an edge between two vertices if they differ in an element of $S$.
This graph is endowed of a set of faces of the form:
\begin{align}
\label{eqn:face_b}
\{x, xb, xb^2, \dots, xb^{\ell-1}, xb^\ell = x\}
\end{align}
and
\begin{align}
\label{eqn:face_ab}
\{x, xa, x(ab), \dots, x(ab)^{m-1}a, x(ab)^m = x\}
\end{align}
for every vertex $x$ of $\tau(\ell, m)$.

\noindent
Remark that the vertex $x$ is also incident to the face:
\begin{align}
\label{eqn:face_ba}
\{x, xb^{-1}, x(b^{-1}a), \dots, x(b^{-1}a)^m = x\}.
\end{align}
This defines an infinite planar tiling $\tau(\ell, m)$.
Figure~\ref{fig:hyperbolic_tiling_zemor} represents one of these tilings around the vertex corresponding to the element 1 in the group.

\begin{figure}[htbp]
\centering
\includegraphics[scale=.8]{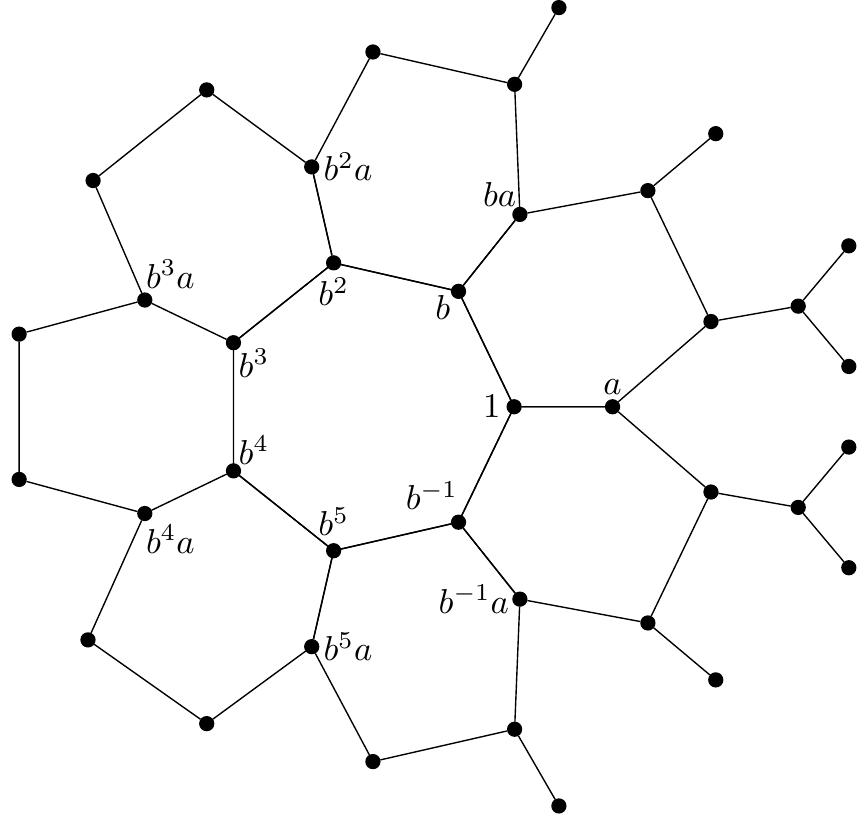}
\caption{Local structure of the tiling $\tau(3, 7)$}
\label{fig:hyperbolic_tiling_zemor}
\end{figure}

To define a color code, we need to construct a finite version of this infinite tiling.
From the infinite square lattice, we can easily identify the opposite edges to obtain a finite version of the square tiling. It becomes more complicated with hyperbolic tilings. In this case, we use the group structure to define a finite tiling. \v{S}ir\'a\v{n}'s strategy is to construct a finite quotient $T_r(\ell, m)$ of the group $T(\ell,m)$ which satisfies exactly the same relations between $a$ and $b$, up to some length $2r+1$. \Siran constructed an infinite family of groups $T_r(\ell, m)$, for $r\in\N$, of cardinality:
\begin{align}
\label{eqn:group_size}
|T_r(\ell, m)| \leq C^r,
\end{align}
for some constant $C=C(\ell, m)$.
From these finite groups, we construct a family of finite tilings $\tau_r(\ell, m)$ with faces defined by the reduction of Equation~(\ref{eqn:face_b}) and Equation~(\ref{eqn:face_ab}).

The surjective morphism from $T(\ell, m)$ to $T_r(\ell, m)$ induces a surjective graph morphism $\pi$ from the infinite graph $\tau(\ell, m)$ to the finite graph $\tau_r(\ell, m)$.
Using this graph morphism, \Siran proved that the graph $\tau_r(\ell, m)$ locally looks like the planar graph $\tau(\ell, m)$. That is:

\begin{prop} \label{prop:tau_r_planaire}
Every ball of radius $r$ of the graph $\tau_r(\ell, m)$ is isomorphic to a ball of same radius in the planar graph $\tau(\ell, m)$.
\end{prop}

To study the minimum distance of surface codes, we will use the fact that a cycle included in such a ball is a cycle of a planar graph, thus it is sum of faces.
In the case of color codes, we will combine this proposition with Lemma~\ref{lemma:color_d_bound} to obtain lower bounds on the minimum distance.

\subsection{Construction of hyperbolic color codes}

When the faces of the graph $\tau_r=\tau_r(\ell, m)$ are not 3-colorable, it cannot be used to define a color code. In this case, we will introduce a cover of the graph $\tau_r$ which is equipped with 3-colorable faces.
For example, the faces of the graph represented in Figure~\ref{fig:hyperbolic_tiling_zemor} is not 3-colorable because it contains a face of even length. In what follows, $\ell$ is always an even integer, so that the faces of $\tau_r$ have even length.

The first step is a characterization of the graphs $\tau_r$ that can be 3-colored. We introduce the subgraph $H_r$ of the dual graph $\tau_r^*$. It is the graph whose vertices are the faces of length $2m$ of $\tau_r$. These faces are described by Equation~(\ref{eqn:face_ab}) and by Equation~(\ref{eqn:face_ba}).
Two vertices of $H_r$ are linked iff the corresponding faces of $\tau_r$ share an edge. This graph $H_r$ is regular of degree $m$.
It leads to the following characterization:

\begin{lemma} \label{lemma:CNS_colorable}
The faces of the graph $\tau_r$ are 3-colorable if and only if the graph $H_r$ is bipartite.
\end{lemma}

Recall that a graph is called bipartite if there exists a partition of the vertex set into two parts such that no edges has both endpoints in one part.

\begin{proof}
Assume that the graph $H_r$ is bipartite. Its vertices can be colored in two different colors, without having the same color on two neighbors. This induces a coloration of the corresponding faces in the graph $\tau_r$. It remains only to color the faces of length $\ell$ defined by the powers of the generator $b$.
By construction, there is only one face of this form around every vertex. Thus, two such faces are never neighbors. To obtain a 3-coloration of the faces, it suffices to color the remaining faces with a third color.
The converse can be proved similarly.
%
\end{proof}

To define a color code, we replace the graph $\tau_r$ by its double cover $\tilde\tau_r$:

\begin{defi} \label{defi:couleur_revetement}
The {\em double cover $\tilde\tau_r$} of the graph $\tau_r=(V, E, F)$ is defined to be the graph whose vertices are the elements of $V \times \F_2$, and whose edges are the pairs $\{(x, i), (y, j)\}$ where $\{x, y\}$ is an edge of $\tau_r$ and $j=i+1$.
Every face $\{x_1, x_2, \dots, x_{2s} \}$ of the graph $\tau_r$ defines two faces of the double cover $\tilde \tau_r$ of the form:
$$\{(x_1,i_0), (x_2, i_0+1), \dots, (x_{2s}, i_0+1) \},$$
with $i_0$ = 0 or 1.
\end{defi}

The next proposition proves that this double cover can be used to define a color code. The row weights of the matrices $\HH_X=\HH_X$ are the length of the faces. We obtain a family of quantum LDPC codes with bounded row weight.

\begin{prop} \label{prop:colorable_cover}
The faces of the graph $\tilde\tau_r$ are 3-colorable.
\end{prop}

\begin{proof}
We adapt the proof of Lemma~\ref{lemma:CNS_colorable}.
We can color in red the isolated faces corresponding to the generator $b$. Then, it remains to prove that the graph $\tilde H_r$, defined like in Lemma~\ref{lemma:CNS_colorable}, is bipartite. To this end, we will transfer the bipartite structure of $\tilde \tau_r$ to the graph $\tilde H_r$.
The vertices of $\tilde H_r$ are the faces of the double cover defined by the power of $ab$:
$$
\{(x,i), (xa,i+1), (x(ab),i), \dots, (x(ab)^{m-1}a,i+1)\}.
$$
This face is denoted by $F(x,i)$.
We will use the index $i \in \F_2$ to partition the vertices of $\tilde H_r$.
Denote by $V_i(\tilde H_r)$ the set of the faces of $\tilde \tau_r$ of the form $F(x, i)$.
These two sets $V_i(\tilde H_r)$, with $i=0$ or $1$, form a partition of the vertices of $\tilde H_r$.
This graph is bipartite because two faces $F(x, i)$ and $F(y, i+1)$ are never neighbors. To prove this, check that if an edge separates two faces of $V(\tilde H_r)$, it is an edge of type $\{(x, 0), (xa, 1)\}$.
\end{proof}

The next proposition proves that the cover $\tilde\tau_r$ locally looks like the original graph $\tau_r$.

\begin{prop} \label{prop:radius_cover}
Every ball of radius $r$ of the double cover $\tilde \tau_r$ is isomorphic to a ball of radius $r$ of $\tau_r$.
\end{prop}

\begin{proof}
Given a vertex $(x, k) \in \tilde V$, we consider the following application:
\begin{align*}
\pi : \tilde V & \longrightarrow V\\
(u, i) & \longmapsto u
\end{align*}
We will prove that the restriction of $\pi$ to a ball of radius $r$ around $x$ is a graph isomorphism. Assume that two different vertices of this ball $(u, i)$ and $(v, j)$ are sent onto the same image $u=v$. We have necessarily $i \neq j$.

There exists a path from $(x, k)$ to $(u, i)$ in the ball of the cover. There exists also a path between $(x, k)$ and $(v, j)$. The concatenation of these two paths is a path from $(u, i)$ to $(v, j)$ passing through $(x, k)$, in the graph $\tilde \tau_r$. The length of this path is odd because the parity of this length is also the parity of $i+j$.
The image of this path under $\pi$ is a cycle included in the ball $B(x, r)$ of the graph $\tau_r$. We can suppose that its length is unchanged. Indeed, if $\pi$ induces an identification of two vertices on these paths before reaching $(u, i)$ and $(v, j)$, we can replace $(x, i)$ and $(y, j)$ by vertices closer to $(x, k)$. Finally, we get a cycle of odd length, included in a ball of radius $r$ of $\tau_r$.
By Proposition~\ref{prop:tau_r_planaire} and the remark below, such a cycle is a sum of faces. As these faces have even length, this cycle should have an even length. This contradiction proves the injectivity of $\pi$, restricted to a ball of radius $r$.
This property does not depend of the vertex $(x, k)$.
\end{proof}

We are now able to estimate the parameters of hyperbolic color codes defined from the covers $\tilde \tau_r$.

\begin{theo} \label{theo:couleur_hyperbolique}
The rate of the hyperbolic color codes associated with $(\tilde \tau_r)_{r\in\N}$ converges to:
$$
2\left( \frac{1}{2}-\frac{1}{\ell}-\frac{1}{m} \right)
$$
and the minimum distance of these codes is at least logarithmic in the length.
\end{theo}

When they are well defined, the color codes associated with $(\tau_r)_r$ have the same rate and the same minimum distance for a length divided by 2.

\begin{proof}
The length of the quantum code is $n=|\tilde V| = 2|V|$.
The number of encoded qubits is $k=4-2\chi(\tilde \tau_r)$.
By construction, the Euler characteristic of the double cover $\tilde\tau_r$ is twice the Euler characteristic of $\tau_r$. Thus, we need to compute $\chi(\tau_r)$. An application of the pigeonhole principle gives: $|E|=(3/2)|V|$ and $|F|=(1/\ell+1/m)|V|$ and as a consequence:
$$
\chi(\tau_r)=\left( \frac{1}{l} + \frac{1}{m} - \frac{1}{2} \right) |V|.
$$
The dimension formula follows:
$$
k(\tilde\tau_r) = 2\left( \frac{1}{2} - \frac{1}{\ell} - \frac{1}{m} \right)(2|V|) + 4.
$$

The minimum distance of the color code associated with $\tilde\tau_r$ is the minimum weight of an element  of $C \backslash C^\perp$.
Let $x$ be a codeword of $C$ and let $F(x)$ be the induced graph defined in Lemma~\ref{lemma:color_d_bound}.
Denote by $M=\max\{\ell, 2m\}$ the maximum length of a face of $\tilde\tau_r$.
The number of vertices of the graph $F(x)$ is at most $3Mw(x)$ because every vertex of $x$ had at most 3 faces to the induced graph.
Without loss of generality, we can assume that $F(x)$ is connected, otherwise, we work on the connected components.
If we have $w(x)<r/(3M)$, then the induced graph $F(x)$ is a connected graph of size $|F(x)|<r$. This graph is necessarily included in a ball of radius $r$ of $\tilde\tau_r$.
From Proposition~\ref{prop:radius_cover}, this ball is isomorphic to a ball of radius $r$ of $\tau_r$, and this ball is planar by Proposition~\ref{prop:tau_r_planaire}.
Therefore, the graph $F(x)$ is planar and Lemma~\ref{lemma:color_d_bound} proves that $x$ is in $C^\perp$.
We proved that if $w(x) < r/(3M)$ then $x$ is in $C^\perp$.
Thus, the minimum distance is at least proportional to $r$.
This quantity is at least logarithmic in the length $n=2|V|=2|T_r|$, by Equation~(\ref{eqn:group_size}).
\end{proof}

\section{Tradeoff based on Gromov's systolic inequalities}
\label{section:gromov}

Hyperbolic surface codes of \cite{FML02} and \cite{Ze09} and our hyperbolic color codes produce topological codes beating Bravyi-Poulin-Terhal bound, that is available only for square tilings. In this section, we obtain a general tradeoff between the parameters of surface codes and color codes based on every tiling. Our strategy is to endow the tiling defining the quantum code with a riemannian metric. This stronger structure allows us to import results from riemannian geometry to the field of quantum codes.

We will use the notion of systole and Gromov's systolic inequalities \cite{Gr92}.
We consider a compact connected surface $\mc V$, endowed with a riemannian metric.
We are interested in the first singular homology group $H_1(\mc V, \F_2)$ of $\mc V$ \cite{Ha02}.
The {\em systole} of $\mc V$ is defined to be the shortest length of a cycle embedded in the surface, which is not homological to zero.
The relation between systole and quantum codes has been investigated by Freedman, Meyer and Luo \cite{FML02}. More recently, Fetaya proved that it is impossible to go beyond a minimum distance in $O(\sqrt n)$ with homological codes \cite{Fe12}.
We will use a result, due to Gromov, which bounds the systole of a surface as a function of the genus and the area of the surface \cite{Gr92}:

\begin{theo} \label{theo:gromov} [Gromov-1992]
The systole $\syst H_1(\mc V, \F_2)$ of an orientable compact connected surface $\mc V$ of genus $g \geq 2$, endowed with a riemannian metric satisfies
\begin{align}
\label{eqn:gromov}
\left( \syst H_1(\mc V, \F_2) \right)^2 \leq C \frac{(\log g)^2}{g} \Area(\mc V),
\end{align}
for some constant $C$.
\end{theo}

Our goal is to apply this result to topological codes.
The discrete notion of systole that appears in the study of surface codes and color codes is the \emph{combinatorial systole} of a tiling of surface $G$, denoted by $\csyst(G)$. It is defined to be the length of the shortest cycle of the graph $G$ that is not a sum of faces.

To apply Gromov's bounds (\ref{eqn:gromov}) to the parameters of a topological code, we will consider the combinatorial tiling of surface $G$ as a topological surface. It can be seen as a smooth surface $\mc V_G$.
We will endow this surface with a riemannian metric that is sufficiently regular on the faces of the tiling.
This is done with the help of a lemma due to Fetaya \cite{Fe12}:

\begin{lemma} \label{lemma:fetaya}
There exists two constants $C_1$ and $C_2$ such that, for every tiling $T$ composed of triangular faces, there exists a riemannian metric on the smooth surface $\mc V_T$ with:
\begin{itemize}
\item the area of every face $f$ of $T$ satisfies $\Area(f) \leq C_1$,
\item $\csyst(T) \leq C_2 \syst H_1(\mc V_T, \F_2)$.
\end{itemize}
\end{lemma}

The first point means that all the faces have approximately the same area. From the second point, the combinatorial systole of the tiling and the systole of the surface are comparable in size.
This lemma combined with the graphical description of surface codes and color codes enables us to apply Gromov's inequalities to surface codes and color codes.

\begin{theo} \label{theo:tradeoff}
Let $G$ be a finite tiling with faces of length at most $m$ and vertices of degree at most $m$. The parameters $[[n, k, d]]$ of the surface code and the color code (when it exists) associated with $G$ satisfy:
$$
k d^2 \leq C (\log k)^2 n,
$$
for a constant $C$ which depends only on $m$.
\end{theo}

\begin{proof}
As recalled in Section~\ref{section:parameters}, the minimum distance of the surface code associated with a tiling $G$ is $\min\{\csyst(G), \csyst(G^*)\}$. Thus, our first task is to bound the combinatorial systole of the graph $G$. Denote by $T$ the triangulation deduced from $G$ by adding a vertex in the middle of every face and by joining it to all the vertices of this face. This triangulation will allow us to apply Fetaya's lemma.
Let $\gamma$ be a cycle of the triangulation $T$. This cycle can be transformed into a cycle of the graph $G$. It suffices to replace a path that goes through the center of a face by a path which follows the boundary of this face. By this transformation, some pairs of edges are replaced by at most $m/2$ edges. Moreover, this transformation conserves the homology class of the cycle. Thus, we have:
$$
\csyst(G) \leq \frac m 4 \csyst(T).
$$
We can now apply Lemma~\ref{lemma:fetaya} and then Theorem~\ref{theo:gromov} to the surface $\mathcal \mc V_T$ defined by the tiling $T$:
\begin{align*}
(\csyst(G) )^2
& \leq ( \frac m 4 C_2 \syst H_1(\mathcal V_T, \F_2) )^2\\
& \leq C' m^2 \frac{(\log g)^2}{g} \Area(\mathcal V_T).\\
\end{align*}
for some constant $C'$.
We can assume that the surface $\mc V_T$ is orientable, otherwise we consider an orientable double cover of $\mc V_T$. If the surface has genus 0 or 1, we apply Bravyi-Poulin-Terhal bound.

The genus of the surface is proportional to $k$. From Lemma~\ref{lemma:fetaya}, the area of the surface, which is the sum of the area of the triangles, is upper bounded by $\Area(\mathcal V_T) \leq C_1 |F(T)|$, where $|F(T)|$ is the number of triangles in the triangulation~$T$. It is exactly twice the number of edges $|E|=n$ of the graph $G$.
Therefore, we obtain the inequality:
\begin{align*}
(\csyst(G) )^2
& \leq C m^2 \frac{(\log k)^2}{k} n.\\
\end{align*}
This result is also available for the systole of the dual graph $G^*$ because we assumed that the degree of vertices of $G$ are bounded by $m$. Finally, we can replace the left side of the inequality by the minimum distance $d$ of the surface code associated to $G$.

The case of color codes is similar. The minimum distance $d$ of a color code is upper bounded by twice the systole of the shrunk lattices \cite{BM06, De12}.
Thus, we need to bound the combinatorial systole of these shrunk lattices. When we start with a tiling of surface with faces of length at most $m$, the faces of the shrunk lattices have length at most $m/2$. Finally, we can proceed as in the case of surface codes.
\end{proof}

\section*{Conclusion}

We constructed a family of color codes based on hyperbolic tilings. This strategy provides a new family of quantum LPDC codes with asymptotically constant rate and growing minimum distance. To the best of our knowledge, it is the first family of color codes with such parameters.

These hyperbolic color codes are clearly different from hyperbolic surface codes. They offer the advantage of being automatically convertible into subsystem color codes \cite{Bo10}.
This leads to a family of subsystem codes with constant rate and growing minimum distance, with weight two gauge operators.
These subsystem codes have the feature of reducing noise during the syndrome measurement. This transformation in subsystem codes is all the more important because such a construction is not known for hyperbolic surface codes.


Using Gromov's systolic inequalities, we proved that the best asymptotic minimum distance of surface codes and color codes with non-vanishing rate is logarithmic in the length. 
Two families of quantum LDPC codes exceed our bound.
Hypergraph product codes of Tillich and Zémor \cite{TZ09} have fixed non-zero rate for a minimum distance that grows as a square root of the length.
Freedman, Meyer and Luo proposed a family of homological codes encoding only one qubit but with a minimum distance in $O( (n \log n)^{1/2})$ \cite{FML02}.

\section*{Acknowledgements}

The author wishes to ackowledge Gilles Zémor for helpful suggestions on this article and Alain Couvreur for carefully reading this paper and providing constructive comments.
This work was supported by the French ANR Defis program under contract ANR-08-EMER-003 (COCQ project).
Part of this work was done while the author was supported by the Délégation Générale pour l'Armement (DGA) and the Centre National de la Recherche Scientifique (CNRS).
He is now supported by the LIX-Qualcomm Postdoctoral fellowship.

\end{document}